\def\draft{0}

\documentclass[a4paper,UKenglish,cleveref, autoref, thm-restate, numberwithinsect, cleveref]{lipics-v2021}

\pdfoutput=1
\nolinenumbers
\hideLIPIcs

\usepackage{amsmath}
\usepackage{amsthm}
\usepackage{amsfonts}
\usepackage{amssymb}
\usepackage{mdframed}
\usepackage{chemmacros}
\usepackage{xcolor}
\usepackage{soul}
\usepackage{svg} 
\usepackage{algorithm}
\usepackage[noend]{algorithmic}

\chemsetup{
  formula = mhchem,
  reactions/tag-open=(R,
  reactions/tag-close=),
}

\usepackage[textsize=tiny,textwidth=0.9\marginparwidth]{todonotes} 

\newcommand{\N}{\mathbb{N}}
\newcommand{\Z}{\mathbb{Z}}
\newcommand{\R}{\mathbb{R}}

\newcommand{\cex}[1]{{\mu(#1)}}
\newcommand{\bcex}[1]{{\bar{\mu}(#1)}}

\renewcommand{\vec}[1]{\mathbf{#1}}
\newcommand{\monomers}{{\bf{\Psi^0}}}
\newcommand{\polymers}{{\bf{\Psi}}}
\DeclareMathAlphabet{\mathpzc}{OT1}{pzc}{m}{it}

\newcommand{\multiset}[1]{\mathpzc{#1}}
\newcommand{\sett}[1]{\mathcal{#1}}

\ifnum\draft=1
\newcommand{\minki}[1]{\textcolor{blue}{[Minki: #1]}}
\newcommand{\DS}[1]{\textcolor{red}{[David: #1]}}

\newcommand{\HA}[1]{\textcolor{olive}{[Hamid: #1]}}
\newcommand{\todoi}[1]{\todo[inline]{#1}}
\else
\newcommand{\minki}[1]{}
\newcommand{\DS}[1]{}

\newcommand{\HA}[1]{}
\renewcommand{\todo}[1]{}
\newcommand{\todoi}[1]{}
\fi

\usepackage{tikz}
\usetikzlibrary{shapes.geometric,shapes.misc}

\title{Computing and Bounding Equilibrium Concentrations in Athermic Chemical Systems} 

\author{Hamidreza Akef}{The University of Texas at Austin, TX, USA}{hakef@utexas.edu}
{https://orcid.org/0009-0001-9253-9737}
{Support was provided by Schmidt Sciences}

\author{Minki Hhan}{The University of Texas at Austin, TX, USA}{minki.hhan@austin.utexas.edu}{https://orcid.org/0000-0001-5143-4587}{Support was provided by Schmidt Sciences}

\author{David Soloveichik}{The University of Texas at Austin, TX, USA \and \url{https://www.solo-group.link/}}{david.soloveichik@utexas.edu}{https://orcid.org/0000-0002-2585-4120}{Support was provided by Schmidt Sciences, Department of Energy award DE-SC0024467, and National Science Foundation SemiSynBio III: GOALI award}

\authorrunning{H. Akef and M. Hhan and D. Soloveichik} 

\Copyright{Hamidreza Akef, Minki Hhan, and David Soloveichik} 

\ccsdesc{Theory of computation~Models of computation}
\ccsdesc{Theory of computation~Design and analysis of algorithms} 

\keywords{Equilibrium concentrations, Thermodynamic Binding Networks, Monomer-polymer model, Detailed balance}

\category{} 

\relatedversion{} 

\acknowledgements{We thank Joshua Petrack for valuable discussions and insights.}

\EventEditors{Josie Schaeffer and Fei Zhang}
\EventNoEds{2}
\EventLongTitle{31st International Conference on DNA Computing and Molecular Programming (DNA 31)}
\EventShortTitle{DNA 31}
\EventAcronym{DNA}
\EventYear{31}
\EventDate{August 25--29, 2025}
\EventLocation{Lyon, France}
\EventLogo{}
\SeriesVolume{347}
\ArticleNo{10}

\begin{document}

\maketitle

\begin{abstract}
Computing equilibrium concentrations of molecular complexes is generally analytically intractable and requires numerical approaches. 
In this work we focus on the polymer-monomer level, where indivisible molecules (monomers) combine to form complexes (polymers). 
Rather than employing free-energy parameters for each polymer, we focus on the athermic setting where all interactions preserve enthalpy. 
This setting aligns with the strongly bonded (domain-based) regime in DNA nanotechnology when strands can bind in different ways, but always with maximum overall bonding---and is consistent with the saturated configurations in the Thermodynamic Binding Networks (TBNs) model. 
Within this context, we develop an iterative algorithm for assigning polymer concentrations to satisfy detailed-balance, where on-target (desired) polymers are in high concentrations and off-target (undesired) polymers are in low.
Even if not directly executed, our algorithm provides effective insights into  upper bounds on concentration of off-target polymers,
connecting combinatorial arguments about discrete configurations such as those in the TBN model to real-valued concentrations.
We conclude with an application of our method to decreasing leak in DNA logic and signal propagation.
Our results offer a new framework for design and verification of equilibrium concentrations when configurations are distinguished by entropic forces.
\end{abstract}

\section{Introduction}
In general, chemical equilibria of complex chemical systems are not analytically solvable and numerical tools are required for analysis.
Such tools include NUPACK~\cite{zadeh2011nupack} for thermodynamic analysis of nucleic-acid systems, as well as more abstract platforms that support domain-level abstraction and free energy specification~\cite{concentratio} including via rule-based modeling~\cite{sekar2016energy},
and software for computing steady-state concentrations of chemical reaction networks~\cite{hoops2006copasi,leal2015reaktoro}.
However, engineers often want a deeper understanding of the equilibrium than what analytically opaque numerical calculations can provide.
Moreover, we often seek to understand infinite classes of designs such as logic circuits constructed from gate modules or parameterized constructions. 
For example, there is a growing body of work on leak in DNA-based systems, exhibiting a family of schemes parameterized by a ``redundancy parameter'' meant to decrease leak arbitrarily at the cost of additional system components~\cite{DBLP:conf/dna/ThachukC15,DBLP:conf/dna/DotyRSTW17,DBLP:journals/pnas/WangB18,DBLP:journals/tcbb/BreikCDHS21,wang2023speed,Wang2024.09.13.612990}. 
It is proven that producing off-target species necessarily decreases the overall number of separate complexes---the change controlled by the redundancy parameter---and thus incurs an entropic penalty.
However, due to system complexity, the relationship between this rigorously proven thermodynamic unfavorability and the actual concentrations of off-target species often remains implicit.

Properties of equilibria of abstract coupled chemical reactions have been extensively explored in the chemical reaction network theory literature~(e.g.,~\cite{feinberg2019foundations}, Chapter 14, for detailed-balance equilibria),
including establishing the conditions on the rate constants and structure of the networks necessary for detailed balance~\cite{feinberg1989necessary,dickenstein2011far}.
Full explicit-parameter schemes have been developed to characterize the entire space of equilibria (e.g.,~\cite{johnston2019deficiency}), but these analytical approaches have severe limitations.  
For large systems the resulting explicit formulas are unwieldy and offer no effective guidance on how to choose the parameters so that off-target concentrations remain below a desired bound.

In this paper 
we are interested in the following problem: 
Our chemical species are complexes (termed polymers) made up of indivisible units called monomers. 
We assume there are finitely many possible polymers,
and among these we are given a set of on-target species that we want to have  desired equilibrium concentrations, and all other off-target species to have some sufficiently low equilibrium concentrations.
Our task is to determine a consistent detailed-balance equilibrium, and therefore the concentrations of the monomers that would lead to this equilibrium.
Note that our interest in setting equilibrium concentrations rather than solving for them (based on initial concentrations or total monomer concentrations, for example) is misaligned with most computational approaches.
For instance, it can be seen as the reverse problem to NUPACK 3's {\tt concentrations} tool which takes the concentrations of the monomers (strands) and returns the corresponding equilibrium concentrations of the polymers (complexes).

While having the flexibility to set monomer concentrations may appear to simplify the problem in the same way that finding \emph{some} detailed-balance equilibrium is easier than computing the one consistent with input monomer concentrations, 
the complexity comes from simultaneously ensuring that off-target polymers remain in low concentration. 
Consider the natural approach of taking logarithms of all concentrations, thereby converting the detailed-balance equations (balancing each reaction) into a linear system amenable to standard linear solvers.
When we fix the concentrations of on-target polymers, the remaining (off-target) concentrations are typically under-determined. 
The linear system then describes an unbounded affine subspace, and there is no obvious way to extract upper bounds on off-target species' concentrations without leaving the linear framework.
Our solution is an iterative algorithm that assigns concentrations to off-target polymers in decreasing order of concentration, ensuring that each off-target species remains below a desired threshold concentration when possible. Importantly, terminating the algorithm at any iteration still provides valid upper bounds for all remaining off-target polymers.

In the most general formulation of the monomer-polymer equilibrium concentration problem, the polymer free energies can be assigned arbitrarily incorporating binding strength, geometric constraints, etc.
In this work, we focus on the simpler \emph{athermic} case rather than tackling the problem in its full generality.
Our allowed polymers are such that all possible reactions between them are enthalpy-neutral.
This model is consistent with systems of strong fully-complementary DNA domains in which domain-level bonds can only switch binding partners but not de-hybridize. 
Thermodynamic Binding Network (TBN)~\cite{DBLP:conf/dna/DotyRSTW17,DBLP:journals/tcs/BreikTHS19} saturated configurations capture this condition, but our setting is more general without a built-in notion of domains (binding sites).

The main result of this paper is \cref{alg1} and \cref{theorem:level-i} showing how starting with desired concentrations of on-target polymers (already in detailed balance), we can set the concentrations of off-target polymers to satisfy detailed balance and thus thermodynamic equilibrium.
In \cref{sec:nontrivial} we explain the apparent difficulties in balancing reactions which our approach needed to overcome.

If, rather than computing exact concentrations of off-target polymers, it is sufficient to bound them, then we refer the reader to \cref{sec:bounding-framework}.
In \cref{sec:TBN-applications} we apply our framework to the analysis of systems in the TBN model specifically, connecting the combinatorial notions of stability and entropy loss in the TBN model to equilibrium concentrations.
In \cref{sec:Examples}, we show applications of our method to the analysis of a simple TBN AND gate, as well as a parameterized family of signal propagation systems (translator cascades) from prior work. 
For the translator cascade, we argue that tuning concentrations of on-target polymers according to our framework is essential for leak to decrease exponentially with the redundancy parameter.
We conclude with a discussion of future work (\cref{sec:discussion}), including a formulation of new combinatorial conditions in the TBN model to make our framework easily applicable.

\section{Model}

Let $\N$ denote the set of nonnegative integers.
Given a finite set $\sett{A}$, we define $\N^\sett{A}$ as the set of functions $f : \sett{A} \to \N$.

A multiset $\multiset{M}$ over the finite set $\sett{A}$ is described by its \emph{counting function} $f_{\multiset{M}} \in \N^\sett{A}$, where for each element $a \in \sett{A}$, the value $f_{\multiset{M}}(a)$ indicates how many times $a$ appears in $\multiset{M}$.
We often write $\multiset{M} \in \N^\sett{A}$ to mean that $\multiset{M}$ is a multiset over $\sett{A}$, and we denote the count of $a \in \sett{A}$ in $\multiset{M}$ by $\multiset{M}[a]$.
The notation $a\in \multiset{M}$ means that $\multiset{M}[a]\ge 1$.
The \emph{cardinality} of a multiset $\multiset{M}$, denoted $|\multiset{M}|$, is the total number of elements in the multiset
$|\multiset{M}|=\sum_{a\in \sett{A}} \multiset{M}[a]$.
For two multisets $\multiset{M}$ and $\multiset{M'}$ over $\sett{A}$, we define their union $\multiset{M} + \multiset{M'}$ as the multiset whose counting function is the pointwise sum
$(\multiset{M}+\multiset{M'})(a) = \multiset{M}[a]+\multiset{M'}[a]
\text{ for all }a \in \sett{A}$.
For example, let $\multiset{M} = \{a, a, b, c\}$. 
Then,
$\multiset{M}[a]=2$, $\multiset{M}[b] = 1$, $\multiset{M}[c] = 1$, and $\multiset{M}[d] = 0$ for all $d \notin \{a, b, c \}$.
Also, the cardinality of $\multiset{M}$ is $|\multiset{M}| = 2 + 1 + 1 = 4$.
Finally, note that $\multiset{M}$ could also be written as the union $\{a, b\} + \{a, c\}$, for example.

The linear combination of multisets with nonnegative integers is defined analogously: For multisets $\multiset{M_1},...,\multiset{M_n}$ and $a_1,...,a_n \in \N$, $\sum_{i=1}^n a_i \cdot \multiset{M_i}$ corresponds to the counting function $\sum_{i=1}^n a_i \cdot   \multiset{M_i}[a]$.
Let $\multiset{M_1}, \multiset{M_2} \in \N^\sett{S}$ be two multisets over the same set $\sett{S}$.
The difference $\multiset{M_1} - \multiset{M_2}$ is defined as the multiset $\multiset{M} \in \N^\sett{S}$ such that for every $a \in \sett{S}$,
$(\multiset{M_1} - \multiset{M_2})[a] = \multiset{M_1}[a] - \multiset{M_2}[a]$
provided that $\multiset{M_1}[a] \ge \multiset{M_2}[a]$ for all $a \in \sett{S}$.

We also define the intersection of a multiset with a \emph{set}.
Given a multiset $\multiset{M}$ over $\sett{A}$ and a subset $\sett{S} \subset \sett{A}$, the intersection $\multiset{M} \cap \sett{S}$ is a \emph{multiset} over $\sett{A}$ defined by
$(\multiset{M} \cap \sett{S})[a] = \multiset{M}[a]$ if $a \in \sett{S}$, otherwise $(\multiset{M} \cap \sett{S})[a]=0$.
For instance, 
$\{a,a,b,c\} \cap \{a,c\} = \{a, a, c\}$.

The main object of this paper is an abstract model of monomers and polymers, motivated by systems in DNA nanotechnology. This model captures how simple indivisible molecules (monomers) combine to form complexes (polymers) under specific physical and chemical constraints.

\begin{definition}\label{def:polymer}
Let $\monomers$ be a finite set of \emph{monomers}, and $\polymers \subseteq \N^{\monomers}$ be a finite set of \emph{polymers} over these monomers, where 
each polymer $P \in  \polymers$ is a multiset of monomers.
\end{definition}
Let $\vec{x^0} \in (0,1)^{\monomers}$ represent the vector of concentrations for all monomers, and let $\vec{x} \in (0,1)^{\polymers}$ represent the vector of concentrations for all polymers (also called configuration).
The relationship between monomer and polymer concentrations is governed by mass conservation. Specifically, we require 
\begin{equation}\label{eqn:constraints}
    \vec{x^0} = \vec{A} \cdot \vec{x} 
\end{equation}
where $\vec{A} \in \N^{|\monomers| \times |\polymers|}$ is a matrix such that each entry $A_{ij}$ specifies the number of monomers of type $i$ in polymer $j$.

For example, the polymer \( P = \{m_1, m_1, m_2, m_3 \} \) contains two copies of \( m_1 \), one of \( m_2 \), and one of \( m_3 \).
Note that we will be interested in cases where the set of polymers of interest $\polymers$ is a finite (proper) subset of all possible polymers over $\monomers$.

In DNA nanotechnology the monomers are typically DNA strands with different sequences. 
Polymers are analogous to a multistranded DNA structure composed of multiple DNA strands.
We use the term polymer rather than ``complex'' in order to better emphasize their composition from monomers and to be consistent with the TBN literature.

To model the equilibrium behavior of such systems, we use the free energy formulation in the notation of Dirks et al.~\cite{DBLP:journals/siam/driksBSWP07},
where the equilibrium concentrations are obtained by minimizing the following free energy function (corresponding to the pseudo-Helmholtz free energy used throughout chemical reaction network theory literature~\cite{horn1972general,feinberg2019foundations}):
\begin{equation}\label{eq:energy_function}
    \mathbf{g}(\vec{x}) = \sum_{P \in \polymers}{x_P (\log x_P - \log \Omega_P - 1)}
\end{equation}
where $x_P$ denotes the concentration of polymer $P$, and $\Omega_P$ is its partition function corresponding to the exponential of the polymer's negative free energy. The minimization is subject to the mass conservation constraint given in \cref{eqn:constraints}.

In this work, we focus on \emph{athermic} systems where all interactions are equally favored enthalpically. 
Thus we assume that $\sum_{P \in \polymers} x_P \cdot \log \Omega_P$ is constant for all configurations $\vec{x}$ satisfying mass conservation (\cref{eqn:constraints}), 
yielding a simplified cost function that is entirely entropic:
\begin{equation}\label{eq:simple_energy}
    g(\vec{x}) = \sum_{P \in \polymers}{x_P (\log x_P - 1)}.
\end{equation}
This function serves as the objective of our optimization problem. Its minimizer under the constraint of \cref{eqn:constraints} corresponds to the equilibrium concentration of polymers.

Understanding the equilibrium requires us to formalize how polymers can transform into one another. We do so by defining reconfigurations and the reactions they induce.

\begin{definition}\label{def:reconfiguration}
    Two multisets of polymers $\multiset{M_1}, \multiset{M_2} \in \N^{\polymers}$ are \emph{reconfigurations} of each other, written $\multiset{M_1} \cong \multiset{M_2}$, if for every monomer $m \in \monomers$, the total count of $m$ is the same in $\multiset{M_1}$ as in $\multiset{M_2}$; i.e., $\sum_{P \in {\polymers}} \multiset{M_1}[P] \cdot P[m] = \sum_{P \in {\polymers}} \multiset{M_2}[P] \cdot P[m]$.
    Whenever $\multiset{M_1} \cong \multiset{M_2}$, we also define \emph{reaction} $\multiset{M_1} \to \multiset{M_2}$.
\end{definition}

We occasionally use Greek alphabets to denote the reactions.
Note that while $\cong$ is a binary relation, a reaction $\alpha:\multiset{M_1} \to \multiset{M_2}$ is an ordered pair representing a directed transformation between the multisets.
Our arguments will refer to \emph{reactants} (left-hand side) and \emph{products} (right-hand side) of particular reactions.\footnote{All reactions are of course reversible, but we treat each direction as a separate reaction.}

An important property of the minimum of \cref{eq:energy_function,eq:simple_energy} is that it satisfies detailed balance over reactions~\cite{horn1972general,feinberg2019foundations}.
For us (\cref{eq:simple_energy}), for any single reaction $\alpha : \multiset{M_1} \to \multiset{M_2}$, the equilibrium concentrations satisfy:
\[
    \prod_{P\in \multiset{M_1}} {x_P}^{\multiset{M_1}[P]} = \prod_{P\in \multiset{M_2}} {x_P}^{\multiset{M_2}[P]}.
\]
We say that a reaction $\alpha$ is balanced when this equality holds. 
This notion of balance will be central to our characterization of equilibrium concentrations formed by on-target polymers, which we explore next.

The fact that balance leads to the minimum of the free energy is well-established, and 
we present its sketch in \cref{appndx:equilibrium} for the completeness of the paper.

\section{Characterizing On-target Polymers}

We now define the set $\sett{S}$ of on-target polymers---intuitively, these are the high-concentration polymers whose equilibrium concentration we set as input to our algorithm.

\begin{definition}\label{def:level0}
    Let $\sett{S} \subseteq \polymers$ be a set of polymers, and let $\mu: \sett{S} \to (0,1]$ be a function assigning a \emph{concentration exponent} to each polymer in $\sett{S}$.
    We say that $\sett{S}$ is \emph{on-target} with concentration exponents $\mu$ if: 
    \begin{enumerate}
        \item For every polymer $P \in \polymers$, there are multisets of polymers $\multiset{M} \in \N^\sett{S}$ and $\multiset{M'} \in \N^{\polymers}$ where $\multiset{M'}$ contains $P$ such that $\multiset{M} \cong \multiset{M'}$.
        \item For any two multisets $\multiset{M_1}, \multiset{M_2} \in \N^\sett{S}$, if $\multiset{M_1} \cong \multiset{M_2}$, then their concentration exponents are equal $\cex{\multiset{M_1}} = \cex{\multiset{M_2}}$.
        Here, the concentration exponent of a multiset $\multiset{M}$ is defined as $\cex{\multiset{M}} = \sum_{P \in \multiset{M}} \multiset{M}[P] \cdot \cex{P}$.
    \end{enumerate}
\end{definition}

While presented in terms of restricting $\sett{S}$, another (perhaps more apropos) interpretation of the first condition is that we are restricting $\polymers$ to be only those polymers that can be obtained via a reconfiguration of polymers over $\sett{S}$.
The interpretation of the second condition is that every reaction among polymers in $\sett{S}$ is in detailed balance. More precisely:
\begin{remark}\label{remark: supposed-concentration}
    Let $0 < c < 1$, and suppose every on-target polymer $P \in \sett{S}$ has concentration $c^\cex{P}$. 
    Then every reaction $\multiset{M_1} \to \multiset{M_2}$, where $\multiset{M_1}, \multiset{M_2} \in \N^\sett{S}$, satisfies $c^\cex{\multiset{M_1}} = c^\cex{\multiset{M_2}}$ and thus the product of the concentrations of the left-hand side polymers is equal to the product of the concentrations of the right-hand side polymers.\footnote{Note that we use symbol $\mu$ for concentration exponents because of the direct parallel to the standard notion of chemical potential, which is expressed with logarithmic terms of concentration: chemical potential $\mu = \mu^\circ + RT \ln(x)$, where $x$ is the mole fraction.}\footnote{Our base concentration $c$ is smaller than one because the units of concentration are mole fractions---the ratio of polymers to solvent molecules, and 
    the derivation of the energy function \cref{eq:energy_function} only holds in the regime of less polymer than solvent~\cite{DBLP:journals/siam/driksBSWP07}.
    } 
\end{remark}

If $\sett{S} = \polymers$ then we are done, being guaranteed that all reactions are in detailed balance. 
However, the case of interest is where we do not know the concentration exponents of all polymers---making it the goal of this paper to compute them or bound them to ensure that the equilibrium concentration of off-target polymers is small.

We now define canonical reactions, which are the reactions with reactants from $\sett S$.
These reactions can generate polymers outside of $\sett{S}$ from polymers in $\sett{S}$.
Each canonical reaction has key quantities we call imbalance and novelty which will play a key role in our algorithm.

\begin{definition}\label{def:canonical}
    For an on-target set $\sett{S}$ with $\mu$, 
    a reaction $\alpha: \multiset{M_1} \to \multiset{M_2}$ is called \emph{canonical} if all its reactants are over $\sett{S}$ (i.e., $\multiset{M_1} \in \N^\sett{S}$).  
    Then the \emph{imbalance} of $\alpha$ is defined as $k(\alpha) := \cex{\multiset{M_1}}-\cex{\multiset{\hat{M_2}}}$, where $\hat{\multiset{M_2}}:=\multiset{M_2} \cap \sett{S}$.
    The \emph{novelty} of $\alpha$ is defined by $l(\alpha) := |\multiset{M_2} - \multiset{\hat{M_2}}|$. 
\end{definition}
Given the definition of multiset subtraction, the novelty $l(\alpha)$ is always non-negative. 

Intuitively, the imbalance of the reaction represents how far from detailed balance it is \emph{if we consider only the polymers whose concentrations are known}.
The novelty is the number of off-target (new) polymers produced by the same reaction.

Intuitively, a large $k(\alpha)$ means a larger bias of the reaction toward its reactants, i.e., elements of $\sett{S}$, prior to assigning concentration exponents to off-target polymers.
This is desired since it implies less pressure to make off-target polymers and gives us more room to maneuver in assigning concentrations to them.

Intuitively, a larger novelty $l(\alpha)$ means that the canonical reaction is more entropically favored to its products. 
Note that the term $\sum_{P \in \polymers}{x_P \log x_P}$ in \cref{eq:simple_energy} is minimized when the concentration is more ``spread out'' over the different species (like Shannon entropy), and thus there is an entropic force to generate new polymers.
A large $l(\alpha)$ is thus undesired.
As justified in the subsequent results, the ratio $k(\alpha)/l(\alpha)$ captures the effective tradeoff between imbalance and novelty.

Note that summing canonical reactions gives another canonical reaction.
While on the one hand this leads to infinitely many canonical reactions,
on the other hand, this additive property allows us to analyze the set of all canonical reactions using a Hilbert basis of \emph{finite size}.
We explain how to employ the Hilbert basis in our algorithm to avoid the infinity of canonical reactions in \cref{subsec:Hilbert}.

The notion of stability of $\sett{S}$, defined below, captures the idea that on-target polymers are in higher concentration compared to off-target polymers. 
Recall that concentration exponents $\mu$ are  $\leq 1$ for polymers in $\sett{S}$ (\cref{def:level0}). 
As shown later by \Cref{theorem:level-i}, 
the following definition ensures that all concentration exponents of polymers outside of $\sett{S}$ are greater than $1$ by our algorithm.
Since the base concentration $c < 1$, this implies that off-target polymers are in smaller concentration.

\begin{definition}\label{def:stable}
    The on-target set $\sett{S}$ with $\mu$ is called \emph{stable} if, for every canonical reaction $\alpha$ with $l(\alpha)\neq 0$, the ratio $k(\alpha) / l(\alpha) > 1$.
\end{definition}

While our formalism allows different concentration exponents for different polymers in $\sett{S}$, nearly all insight can be obtained from the simple case of uniform on-target concentration exponents:

\begin{definition}\label{def:uniform}
    The on-target set $\sett{S}$ with $\mu$ is called 
    \emph{uniform} if every polymer $P \in \sett{S}$ has concentration exponent $\cex{P} = 1$.
\end{definition}

\section{Why Balancing is Nontrivial} \label{sec:nontrivial}
Our main goal is to ensure that in equilibrium the concentration of on-target polymers is much higher than the others. 
This section presents two examples illustrating why this is nontrivial.

Consider a uniform set of on-target polymers $\sett S = \{A,B,C,D\}$, aiming at the equilibrium concentration $x_A=x_B=x_C=x_D=c$. For two off-target polymers $P_1$ and $P_2$, suppose that we figured out a canonical reaction
\[    
    \alpha:2A+B+C \to P_1 + P_2.
\]
This reaction is in detailed balance if and only if $x_{P_1} \cdot x_{P_2}=x_A^2 \cdot x_B \cdot x_C = c^4.$ 
At this point, it is unclear how to balance $x_{P_1}$ and $x_{P_2}$ without inspecting the other reactions:
For example, the reaction $\alpha$ is balanced by assigning $x_{P_1}=x_{P_2}=c^{2}$, 
but another canonical reaction $A+2B+2D \to 2P_1$ would not be if it exists.

The following example exhibits a different potential problem.
While \cref{remark: supposed-concentration} shows that reactions entirely within $\sett{S}$ are balanced, it is not clear that given our choice of concentration exponents $\mu$ for $\sett{S}$, the reactions involving polymers outside of $\sett{S}$ could be balanced at all.
Suppose there are two reactions
\[
    \beta:A+B \to P_3 \text{ and }\gamma:B+C+D \to 2P_3.
\]
These reactions intuitively say that the off-target polymer $P_3$ is non-interacting to other polymers in these reactions so that the above problem of balancing the concentrations over off-target polymers is absent.\footnote{Looking ahead, these non-interacting polymers are indeed easier to assign concentrations to as shown in \cref{subsec:non-interacting}.}
However, balancing the reaction $\beta$ suggests the concentration $x_{P_3} = c^2$ at equilibrium, but the reaction $\gamma$ gives $x_{P_3} = c^{1.5}$.
In other words, it is unclear if the concentrations of each of the polymers in $\sett S$ can be the same or even close to each other.

Despite these difficulties, our main result guarantees that the configuration on on-target polymers we demand is in equilibrium without the above issue, and 
can be \emph{extended} to the configuration over all polymers at equilibrium, and ensures that the off-target polymers remain in very low concentration at equilibrium.
Intuitively, we prove that there exists a canonical reaction for which we can assign concentrations to product polymers without creating conflicts with other reactions, and we provide a method to find this reaction.
Moreover, as we will see later, these assignments occur in order of decreasing concentration allowing us to restrict the concentration of off-target species.

\section{Main Result: Concentration of Off-target Polymers}
For the remainder, we fix $\monomers$, $\polymers$, and a particular set of on-target polymers $\sett S$ with concentration exponents $\mu$.

To systematically assign concentration exponents to all polymers, we will organize them into hierarchical groups called levels. The process begins with a designated set of polymers $\sett S$, the on-target polymers, which are assigned initial concentration exponents via $\mu: \sett S \to \mathbb{R}^+$.

All other polymers, called off-target polymers, will be partitioned into level sets $\sett S_1, \sett S_2, \ldots$, constructed inductively based on how these polymers appear in certain canonical reactions. At each level $i$, we will compute a scalar value $\mu_i$ that serves as the concentration exponent for all polymers newly added at that level.

In this way, we gradually extend the initial function $\mu$ to a global function $\bar\mu: \polymers \to \mathbb{R}^+$ that assigns a concentration exponent to every polymer in the system. This inductive process and the precise definitions of $\mu_i$, $\sett S_i$, and $\bar\mu$ will be described in detail below.

\begin{definition}\label{def:imbalance_novelty}
    Given $\sett S_0,...,\sett S_{i-1}$, assume $\bcex{P}$ is assigned for any $P \in \bigcup_{j=0}^{i-1} \sett S_{j}$.
    For a canonical reaction $\alpha:\multiset{M_1} \to \multiset{M_2}$, 
    we define $\multiset{\hat{M_2}}:= \multiset{M_2} \cap\left(\bigcup_{j=0}^{i-1} \sett S_{j}\right)$. 
    The $i$th-level \emph{imbalance} of $\alpha$ is defined as $k_i(\alpha) := \bcex{\multiset{M_1}}-\bcex{\multiset{\hat{M_2}}}$, and
    the $i$th-level \emph{novelty} by $l_i(\alpha) := |\multiset{M_2}| -|\multiset{\hat{M_2}}|$.
\end{definition}
While $\multiset{\hat{M_2}}$ technically depends on the level index $i$, we suppress this dependency in the notation for simplicity; it will be clear from the context that it is updated at each level. 
Additionally, by the definition of a canonical reaction, we have $\bcex{\multiset{M_1}} = \cex{\multiset{M_1}}$ in the expression above.

\begin{definition}\label{def:ratio}
    Let $\mu_i = \min_\alpha \{k_i(\alpha)/l_i(\alpha)\}$ where the minimum is taken over all canonical reactions $\alpha$ with $l_i(\alpha)\neq 0$.\footnote{The minimum can actually be taken over a finite subset of canonical reactions using Hilbert basis. We refer the reader to \cref{subsec:Hilbert} for more detail.}
    The canonical reactions achieving the minimum are termed \emph{$i$-levelizing} canonical reactions. 
    The \emph{$i$th level set} $\sett S_i$ is defined as the set of all polymers $P$ appearing in any $i$-levelizing reactions not already in $\bigcup_{j=0}^{i-1} \sett S_{j}$.
    Every polymer $P \in \sett S_i$ is assigned concentration exponent $\bcex{P} := \mu_i$.
\end{definition}

At the heart of this inductive construction is the requirement that each $i$-levelizing reaction maintains balance with respect to the assigned concentration exponents. That is, for every reaction $\alpha: \multiset{M_1} \to \multiset{M_2}$ used to define level $\sett S_i$, the assigned exponents ensure the reaction remains detailed balanced, $c^{\bcex{\multiset{M_1}}} = c^{\bcex{\multiset{M_2}}}$.
This holds because each polymer in $\multiset{M_2}$ that already appeared in previous levels contributes its already-defined exponent, while newly introduced polymers in $\sett S_i$ all receive the same value $\mu_i$. Decomposing $\multiset{M_2}$ into previously levelized polymers $\multiset{\hat M_2}$ and new polymers in $\sett S_i$, we obtain:
    \[
        {\bcex{\multiset{M_1}}} 
        = \bcex{\multiset{\hat{M}_2}} + k_i(\alpha)
        = \bcex{\multiset{\hat{M}_2}} + \mu_i \cdot l_i(\alpha) 
        =\bcex{\multiset{\hat{M}_2}} + \mu_i\cdot (|\multiset{M_2}| - |\multiset{\hat M_2}|)
        = \bcex{\multiset{M_2}},
    \]
ensuring that the total concentration exponent on both sides of the reaction is equal. This confirms that the assignment of $\mu_i$ for polymers in level $\sett S_i$ is consistent with the balance requirements dictated by the underlying reactions.

The full procedure for determining the level sets $\sett S_i$ and the corresponding concentration exponents $\mu_i$ is formally specified in \cref{def:imbalance_novelty,def:ratio} and carried out through the step-by-step process described in \cref{alg1}. This algorithm inductively builds the extended concentration exponent function $\bar\mu$ by identifying polymers that can be levelized at each step and assigning them an appropriate exponent value based on their participation in canonical reactions.

While the algorithm itself does not explicitly state a stopping condition, its termination is ensured by the structure of the level construction process. Each time a new level $\sett S_i$ is defined, it includes at least one polymer not present in any previous level. Since the set of all polymers $\polymers$ is finite by assumption,
only finitely many levels can be introduced. As a result, the inductive process must terminate after assigning a level and concentration exponent to each polymer, thus completing the definition of the extended function $\bar\mu$.

For us, canonical reactions represent the simplest meaningful interactions and serve as building blocks for more complex behavior. 
Their central role for our results is motivated by the following lemma:

\begin{lemma}\label{lem:any_to_canonical}
    Let $\vec {x^0} \in (0,1)^{\monomers}$ be a vector of concentrations of the monomers.
    If all canonical reactions are balanced at configuration $\vec x \in (0,1)^{\polymers}$,
    then the cost function $g(\vec x)$ is minimum subject to $\vec A \cdot \vec x =\vec {x^0}$.
\end{lemma}
\begin{proof}
    We will show that any arbitrary reaction can be decomposed into a canonical reaction. Consequently, if detailed balance holds for all canonical reactions, it follows that detailed balance holds for all reactions. 
    Per \cref{appndx:equilibrium}, the cost function $g$ reaches its minimum---under the constraint $\vec A \cdot \vec x =\vec {x^0}$---when all reactions are balanced. This confirms the statement of the lemma.

    Consider an arbitrary non-canonical reaction $\alpha:\multiset{M_1} + P\to \multiset{M_2}$ with $P \notin \sett{S}$ where $\multiset{M_1} + P$ denotes the union of two multisets $\multiset{M_1}$ and $\{P\}$.
    According to the first condition of \cref{def:level0}, there exists a canonical reaction $\beta:\multiset{M_1'}\to \multiset{M_2'} + P$ that produces $P$.
    Now, apply $\beta$ and $\alpha$ sequentially on the reactants $\multiset{M_1} + \multiset{M_1'}$. This gives rise to a new reaction:
    \[\gamma:\multiset{M_1} + \multiset{M_1'}
    \to  \multiset{M_1} + (\multiset{M_2'} + P)
    = (\multiset{M_1}+ P) + \multiset{M_2'}
    \to  \multiset{M_2} + \multiset{M_2'}.\]
    Therefore, the reaction $\alpha$, when catalyzed by $\multiset{M_2'}$ (i.e., adding $\multiset{M_2'}$ to reactants and products) and using the inverse of the canonical reaction $\beta$, can be replaced by the reaction $\gamma$, which involves fewer reactant polymers outside of $\sett{S}$ than $\alpha$ did originally:
    \[
    (\multiset{M_1} + P)+ \multiset{M_2'} 
    = \multiset{M_1}  + (\multiset{M_2'}+ P)
    \to
    \multiset{M_1} + \multiset{M_1'}
    \to
    \multiset{M_2} + \multiset{M_2'}.
    \]
    By repeating this procedure, $\alpha$ decomposes into a canonical reaction.
    This concludes the proof.
\end{proof}

With this groundwork in place, we now state our main theorem, which characterizes the equilibrium distribution of the entire polymer system, both on-target and off-target, in terms of their assigned levels.

\begin{theorem}\label{theorem:level-i}
    Let $\sett S$ be the stable set of on-target polymers with concentration exponents $\mu: \sett{S} \to (0,1]$. For the extended concentration exponent $\bar \mu:\polymers \to \R^+$ generated by \cref{alg1} and 
	for any $0 < c < 1$, there are monomer concentrations $\vec{x^0} \in (0,1)^{\monomers}$ such that the configuration $\vec{x} \in (0,1)^{\polymers}$ with each polymer $P \in \polymers$ at concentration $c^{\bcex{P}}$ is the minimum of $g(\vec{x})$ subject to $\vec{A} \cdot \vec{x} = \vec{x^0}$ (i.e., \cref{eq:simple_energy,eqn:constraints}).
    {Furthermore, $\bcex{P}\ge \mu_1 >1$ for all $P \notin \sett S$.} 
\end{theorem}

In this configuration, every off-target polymer has strictly lower concentration than any on-target polymer. This is because concentrations scale exponentially with the extended exponent $\bcex{P}$, and off-target polymers are assigned strictly larger exponents than the shared minimal value $\mu$ of the on-target set $\sett S$. As a result, for $0 < c < 1$, off-target concentrations are exponentially suppressed. 

\begin{algorithm}
\caption{Calculating level sets and concentration exponents. The repeat loop will terminate because there are finitely many polymers. 
The Hilbert basis implementation avoiding the infinite set $\Lambda$ is discussed in \cref{subsec:Hilbert}.
}
\label{alg1}
\begin{algorithmic}[1]
    \STATE \textbf{Input:} A set of level-0 polymers $\sett S_0=\sett S$ with $\cex{\cdot}$ and all canonical reactions $\Lambda$
    \STATE \textbf{Output:} Sets of level-$i$ polymers $\sett S_i$ and concentration exponents $\bcex{\cdot}$
    \STATE Set $i\gets 0$, $\sett S_j\gets \emptyset$ for all $j\in \N$
    \REPEAT
        \STATE $i\gets i+1$
        \FOR{each canonical reaction $\alpha: \multiset{M_1} \to \multiset{M_2} \in \Lambda$}
            \STATE $\multiset{\hat M_2}\gets$ the multiset of all level $\leq (i-1)$ polymers in $\multiset{M_2}$.
            \STATE $k_i(\alpha) \gets \bcex{\multiset{M_1}}-\bcex{\multiset{\hat M_2}}$ and $l_i(\alpha) \gets |\multiset{M_2}| - |\multiset{\hat M_2}|$
        \ENDFOR
        \STATE Compute $\mu_i = \min_\alpha \{k_i(\alpha)/l_i(\alpha)\}$ where $\min$ is taken over $\alpha\in \Lambda$ with $l_i(\alpha)\neq 0$.
        \FOR{each $\alpha:\multiset{M_1} \to \multiset{M_2} \in \Lambda$ such that $\mu_i=k_i(\alpha)/l_i(\alpha)$}
            \STATE Append $\alpha$ to the set of $i$-levelizing reactions 
            \STATE Append all polymers $P$ in $\multiset{M_2}$ that are not in $\multiset{\hat M_2}$ to $\sett S_i$ and assigns $\bcex{P}=\mu_i$
        \ENDFOR
    \UNTIL {All polymers are included in $\cup_{j=0}^i \sett S_j$}
    \STATE \textbf{return} The sets $\sett S_i$ and the concentration exponents $\bcex{\cdot}$
\end{algorithmic}
\end{algorithm}

\begin{proof}[Proof of \cref{theorem:level-i}]
    We use a proof by contradiction to show that each canonical reaction must be balanced in the configuration with $[P]=c^{\bcex{P}}$, 
    which suffices to conclude the proof thanks to \cref{lem:any_to_canonical}.

    Suppose that there exists a canonical reaction $\alpha:\multiset{M_1} \to \multiset{M_2}$.
    We consider two cases $\bcex{\multiset{M_1}}<\bcex{\multiset{M_2}}$ or $\bcex{\multiset{M_1}}>\bcex{\multiset{M_2}}$. Recall all polymers are levelized.
    
    \noindent{\bf Case 1:} $\bcex{\multiset{M_1}}<\bcex{\multiset{M_2}}.$ 
    Let $t$ be the top level among the polymers in $\multiset{M_2}$.
    Note that $\bcex{P}=\mu_t$ for each $P\in \sett S_t$.
    Since $\mu_t$ is defined as the minimum of $k_t(\alpha') / l_t(\alpha')$ over all canonical reactions $\alpha' \in \Lambda$ with $l_t(\alpha') \ne 0$, we must have: 
    \[
    \mu_t \le \frac{k_t(\alpha)}{l_t(\alpha)} = \frac{\bcex{\multiset{M_1}} - \bcex{\multiset{\hat M_2}}}{l_t(\alpha)} <\frac{\bcex{\multiset{M_2}} - \bcex{\multiset{\hat M_2}}}{l_t(\alpha)} = \mu_t
    \]
    which is a contradiction.

    \noindent{\bf Case 2:} $\bcex{\multiset{M_1}}>\bcex{\multiset{M_2}}.$ 
    Let $\beta_P$ be the levelizing reaction including $P$ as product for each polymer $P \in \multiset{M_2}$.
    Consider a canonical reaction 
    $\sum_{P} \multiset {M_2}[P]\cdot \beta_P :\multiset{M_1'} \to \multiset{M_2'} + \multiset{M_2}$ obtained by summing $\multiset{M_2}[P]$ copies of $\beta_P$
    which includes $\multiset{M_2}$ as products.\footnote{Here we use the linear combination of reactions in a standard way; for example, for $\alpha:\multiset {N_1} \to \multiset {N_2}$ and $\beta: \multiset {N_1'} \to \multiset {N_2'}$, the reaction $2 \cdot \alpha + \beta$ denotes the reaction $2\multiset{ N_1}  + \multiset {N_1'} \to 2 \multiset {N_2} + \multiset {N_2'}$. }
    By applying $\multiset{M_2} \to \multiset{M_1}$ to this reaction, we have a canonical reaction $\beta: \multiset{M_1'} \to \multiset{M_2'} + \multiset{M_2} \to \multiset{M_2'} +\multiset{M_1}$ such that
    \[
    \bcex{\multiset{M_1'}} =
    \bcex{\multiset{M_2'}}+ \bcex{\multiset{M_2}}  
    < \bcex{\multiset{M_2'}} + \bcex{\multiset{M_1}}.
    \]
    Therefore this case is reduced to {\bf Case 1}, leading once again to a contradiction.

    The ``Furthermore'' part is proven in \cref{cor:minbcex}.
\end{proof}

\subsection{Non-Interacting Off-target Polymers}\label{subsec:non-interacting}
Our main theorem implies that if the on-target polymers $\sett S$ and the corresponding concentration exponents $\mu$ are chosen appropriately, the equilibrium \emph{exists} with the extended concentration exponents $\bar \mu$ consistent with $\mu$. 
In this section, 
we show that $\bar \mu$ exponents for \emph{non-interacting} off-target polymers $P$ can be assigned in a particularly easy way. 
Non-interacting off-target polymers are those that can be independently made by a canonical reaction:
\begin{corollary}
    Let $\sett S$ be a stable set of on-target polymers.
    Suppose that an off-target polymer $P\notin\sett S$ is a product of a canonical reaction $\alpha_P:\multiset{M}_{P} \to \multiset{M_{\it P}'} + P$ for some multisets $\multiset{M}_{P} $ and $ \multiset{M_{\it P}'}$ from $\sett S$.\footnote{Note that there may be other reactions involving $P$ with other off-target polymers.} Then $\bcex{P} = \cex{\multiset{M}_{P}} - \cex{\multiset{M_{\it P}'}}.$
\end{corollary}
\begin{proof}
    By \cref{theorem:level-i}, the reaction $\alpha_P$ is in detailed balance, $c^{\bcex{\multiset{M}_{P}}} = c^{\bcex{\multiset{M_{\it P}'} + P}}$, and we take the logarithm of both sides.
\end{proof}

The same idea extends further.
At any point of the execution of the algorithm, suppose that the extended concentration exponents of the set of polymers $\bar {\sett S}$ have been already assigned.
For a non-interacting polymer $P$ outside of $\bar {\sett S}$ with the reaction $\alpha_P:\multiset{M}_{P} \to \multiset{M_{\it P}'} + P$ for multisets $\multiset{M}_{P} $ and $ \multiset{M_{\it P}'}$ from $\bar {\sett S}$, we can assign $\bcex{P} = \bcex{\multiset{M}_{P}} - \bcex{\multiset{M_{\it P}'}}.$ 
This assignment is valid by the same reason to the corollary.

\begin{example}
    Consider the set of monomers and polymers given by $\monomers=\{a,b,c\}$ and $\polymers=\{A,B,C,X,Y,Z\}$; where $A=\{a,a\}$, $B=\{a,b\}$, $C=\{c\}$, $X=\{a,a,a,b\}$, $Y=\{b,b,c,c\}$, and $Z=\{b,b,c,c,c\}$.
    Consider the uniform on-target set $\sett S=\{A,B,C\}$.
    Instead of inspecting all canonical reactions,
    we can focus on three canonical reactions
    \[
        \alpha:A+B \to X,
        ~~\beta:3B+2C\to X+Y,~~~\gamma: A+2B+3C\to X+Z.
    \]
    Here $X$ is the non-interacting off-target polymer in $\alpha$, thus $\bcex{X}=2$. After assigning $\bcex{X}$, $Y$ and $Z$ become non-interacting off-target polymers in $\beta$ and $\gamma$, respectively, so that we derive $\bcex{Y}=3$ and $\bcex{Z}=4$.
\end{example}

\section{Framework for Upper-Bounding Off-target Polymers}
\label{sec:bounding-framework}

Often we are only interested in an upper bound on the concentration of an off-target undesired polymer.
In this case, rather than generating its exact equilibrium concentration via \cref{alg1}, 
we can more efficiently compute an upper bound by narrowing our focus to reactions that directly produce $P$ instead of examining the full set of canonical reactions. 
This leads to a simpler surrogate quantity that approximates $\bcex P$ from below.

\begin{definition}\label{def:heuristic}
    For $P\notin \cup_{j=0}^{i-1} \sett S_j$, let $\widetilde\mu_i(P) = \min_{\alpha}\{k_i(\alpha)/l_i(\alpha)\}$ where the minimum is taken over all canonical reactions $\alpha$ that include $P$ as a product.
\end{definition}

Since $P$ has not yet been levelized by construction, any such reaction $\alpha$ producing $P$, obviously involves at least one unlevelized polymer, ensuring $l_i(\alpha) \neq 0$.

Intuitively, $\widetilde\mu_i(P)$ captures a conservative estimate of the exponent with which polymer $P$ can grow in concentration at level $i$,
based only on reactions that actually produce $P$. Since it is defined as a minimum over a subset of possible reactions, it is easier to compute than the full concentration exponent $\bcex{P}$, yet is still meaningful for analysis:

\begin{theorem}\label{theorem:heuristic}
    For any polymer $P\notin \cup_{j=0}^{i-1} \sett S_j$,
    $\bcex{P} \ge \widetilde\mu_i(P) \ge \mu_i$.
\end{theorem}

\begin{lemma}\label{lemma:nondecreasing_ratio}
    For a canonical reaction $\alpha$ with $l_{i+1}(\alpha)\neq 0$, it holds that
    $k_i(\alpha)/l_i(\alpha) \le k_{i+1}(\alpha)/l_{i+1}(\alpha)$.
\end{lemma}

\begin{proof}
    Note that $\mu_i \le k_i(\alpha)/l_i(\alpha)$ by definition of $\mu_i$. Let $n$ be the count of level-$i$ product polymers in $\alpha$, which must be smaller than $l_{i+1}(\alpha)$. Then we have
    \[
    \frac{k_{i+1}(\alpha)}{l_{i+1}(\alpha)} 
    = \frac{k_{i}(\alpha) - n \mu_i}{ l_i(\alpha) - n} 
    = \frac{k_{i}(\alpha) - n \mu_i}{ l_i(\alpha) - n} \ge \frac{k_{i}(\alpha) - n (k_i(\alpha)/l_i(\alpha))}{ l_i(\alpha) - n} = \frac{k_i(\alpha)}{l_i(\alpha)}.
    \]
\end{proof}

\begin{proof}[Proof of \cref{theorem:heuristic}]
    Since all polymers will eventually be levelized, there exists a $j$-levelizing canonical reaction $\alpha$ that includes $P$ as a product for some $j \ge i$. Then,
    $\widetilde \mu_i(P) \le k_i(\alpha)/l_i(\alpha) \le k_j(\alpha)/l_j(\alpha) = \bcex{P}$
    where the first inequality follows from \cref{def:heuristic}, and the second follows from \cref{lemma:nondecreasing_ratio}.

    Furthermore, the value $\widetilde\mu_i(P)$ is computed by taking the minimum of $k_i(\alpha)/l_i(\alpha)$ only over canonical reactions that produce $P$ (and have $l_i(\alpha) \ne 0$). In contrast, $\mu_i$ is defined as the minimum over all canonical reactions, regardless of whether they produce $P$ or not.
    In other words, the set of reactions considered when computing $\widetilde\mu_i(P)$ is a subset of those considered for $\mu_i$, making the search space for $\widetilde\mu_i(P)$ more restricted. Since a minimum over a smaller set cannot be smaller than the minimum over a larger set, we conclude that $\widetilde\mu_i(P) \ge \mu_i$.
\end{proof}

To summarize, $\widetilde\mu_i(P)$ serves as a computationally efficient lower bound on $\bcex{P}$.
When used alongside \cref{theorem:level-i}, this upper bounds the equilibrium concentration of $P$.

    The usefulness of \cref{theorem:heuristic} extends beyond individual estimates. It contributes to structural insights about level assignments during the iterations of our algorithm. Specifically, after levelizing up through $S_{i-1}$, we know that any unlevelized polymer $P \notin \cup_{j=0}^{i-1} \sett S_j$ satisfies $\bcex{P} \ge \mu_i > \mu_{i-1}$. This inequality ensures that polymers awaiting level assignment must exhibit smaller concentrations.
    In particular, because the first level imbalance $k_1$ and novelty $l_1$ are precisely identical to $k$ and $l$ defined in \cref{def:canonical},
    we derive 
    $\bcex{P} \ge \mu_i \ge \min_{\alpha}\{k(\alpha)/l(\alpha)\}>1$ and the following corollary.
    \begin{corollary}\label{cor:minbcex}
        $\bcex{P}>1$ for all $P \notin \sett S$.
    \end{corollary}

\section{Concentrations of Polymers in the TBN Model} \label{sec:TBN-applications}

While the TBN model is combinatorial in nature, 
quantifying over discrete (saturated) configurations, 
at the end we are most often interested in determining real-valued concentrations, which are accessible to (bulk) wet-lab experiment.
The framework developed in this paper helps to bridge this gap between combinatorics of discrete configurations and concentrations.\footnote{Of course, much of the heavy lifting in bridging this gap is done by the derivation~\cite{DBLP:journals/siam/driksBSWP07} of the cost function $g(\vec{x})$, but our work expands on it beyond numerical simulation.}

Consider a saturated (i.e., maximally bonded) configuration $\multiset{M}$ in the TBN model. 
Quantified over all saturated reconfigurations $\multiset{M'} \cong \multiset{M}$, 
the key quantity of interest in the TBN model is the ``entropy'' of $\multiset{M'}$, defined as the number of polymers in $\multiset{M'}$ (i.e., $|\multiset{M'}|$).
The TBN model defines the ``stable'' configurations to be those that have maximum entropy among all saturated configurations.

Corresponding to the TBN literature~\cite{DBLP:conf/dna/DotyRSTW17}, a multiset of polymers $\multiset{M}$ is called \emph{TBN-stable} if any reaction $\multiset{M} \to \multiset{M'}$ has $|\multiset{M}| \geq |\multiset{M'}|$.
We say that a set $\sett{S}$ of polymers is \emph{TBN-stability closed} if every multiset $\multiset{M} \in \N^\sett{S}$ is TBN-stable and further any reaction $\multiset{M} \to \multiset{M'}$ where $\multiset{M'}$ contains a polymer outside of $\sett{S}$ (i.e., $P \in \multiset{M'}$ for some $P \notin \sett S$) satisfies   $|\multiset{M}| > |\multiset{M'}|$.
In other words, producing a polymer outside of a TBN-stability closed set $\sett{S}$ costs entropy.
The following lemma connects our notion of stability of on-target polymers (\cref{def:stable}) to TBN-stability.

\begin{lemma}\label{lem:TBN-stable}
    If $\sett S$ is TBN-stability closed and satisfies condition (1) of \Cref{def:level0},
    then $\sett{S}$ is on-target with $\mu(P)=1$ for all $P\in \sett S$ (i.e., uniform) and stable.
\end{lemma}
\begin{proof}
    We first check $\sett S$ is on-target with $\mu$.
    For a reaction $\alpha: \multiset{M_1} \to \multiset{M_2}$ where  $\multiset{M_1}$ and $\multiset{M_2}$ are multisets over $\sett S$, $\cex{\multiset{M_1}} = |\multiset{M_1}| \ge |\multiset{M_2}|=\cex{\multiset{M_2}}$. 
    Applying the same argument to $\alpha': \multiset{M_2} \to \multiset{M_1}$
    yields $\cex{\multiset{M_1}} \le\cex{\multiset{M_2}}$. Combining the two gives $\cex{\multiset{M_1}} =\cex{\multiset{M_2}}$, proving $\sett S$ with $\mu$ is an on-target set.

    Now we will prove that uniform on-target $\sett S$ is stable.
    Consider a canonical reaction $\beta: \multiset{M_1'} \to \multiset{M_2'}$ with $l(\beta)>0$. By definition, we have $k(\beta) = \cex{\multiset{M_1'}} - \cex{\multiset{\hat M_2'}} = |\multiset{M_1'}| - |\multiset{\hat M_2'}|$
    and 
    $l(\beta) = |\multiset{M_2'}| - |\multiset{\hat M_2'}|$. 
    Therefore, the condition for $\sett{S}$ being stable, namely $k(\beta)/l(\beta) > 1$, is implied by $|\multiset{M_1'}| > |\multiset{M_2'}|$.
\end{proof}

Thus, the polymers in the TBN-stability closed set $\sett{S}$ represent the intended “high-concentration” polymers, while everything outside of $\sett{S}$ is considered undesired (off-target).

Let canonical reaction $\alpha$ be $\multiset{M} \to \multiset{M'}$; we define $|\multiset{M}| - |\multiset{M'}|$ to be the \emph{entropy loss} $e(\alpha)$ of the reaction $\alpha$.
Recall that during the first iteration of our algorithm,
novelty $l(\alpha)$ is the number of off-target polymers generated in canonical reaction $\alpha$.
The imbalance $k(\alpha)$ in the first iteration can be understood in terms of the entropy loss of the reaction:
the decrease in the number of polymers of a reaction is exactly $e(\alpha) = k(\alpha) - l(\alpha)$.
Thus to have an upper bound on the concentration of off-target polymers via \cref{theorem:heuristic}, 
it is sufficient to find the smallest ratio $e(\alpha)/l(\alpha)$ of any reaction:
\begin{corollary} \label{cor:user-friendly}
Let set $\sett{S} \subseteq \polymers$ be a TBN-stability closed set of polymers.
Let $\mu_1 = \min_\alpha \{e(\alpha)/l(\alpha)\} + 1$ minimized over all reactions $\alpha: \multiset{M} \to \multiset{M'}$ where $\multiset{M} \in \mathbb{N}^\sett{S}$, $\multiset{M'} \in \mathbb{N}^\polymers$ and $l(\alpha) > 0$.
For any $0 < c < 1$, there are monomer concentrations $\vec{x^0} \in (0,1)^{\monomers}$ and configuration $\vec{x} \in (0,1)^{\polymers}$ that minimizes $g(\vec{x})$ subject to $\vec{A} \cdot \vec{x} = \vec{x^0}$ where $\vec{x}$ satisfies:
each polymer $P \in \sett S$ has concentration exactly $c$,
and each polymer $P \in \polymers \setminus \sett S$ has concentration not more than $c^{\mu_1}$.
\end{corollary}

In other words, if $\sett{S}$ is a TBN-stability closed set of polymers then we can assign concentration $c$ to each of them (uniform assignment). 
Then we consider the ``worst'' way to generate any polymers outside of $\sett{S}$ (i.e., off-target): 
the way that has the smallest entropy loss and the largest novelty.
The ratio of the entropy loss to novelty gives us an upper bound on the concentration exponent $\mu_1$ of any off-target polymer, bounding its concentration by $c^{\mu_1}$.
The smallest entropy loss to generate off-target polymers is already the key item of interest in the TBN literature.
Thus the above corollary helps to bootstrap concentration bounding arguments via existing arguments based on entropy loss.

\section{Example Applications}\label{sec:Examples}
In this section we show several applications of our mathematical tools in the analysis of existing systems of interest in the DNA molecular programming literature.
We base our arguments on previous results proving the entropy loss (quantity $e(\alpha)$ of \Cref{cor:user-friendly}) of the systems in producing off-target (leak) species. 

We note that while existing literature succeeds in characterizing entropy loss via TBN-like combinatorial arguments, more work is needed to develop similarly rigorous combinatorial arguments on novelty (quantity $l(\alpha)$); see also Discussion.
In the examples to follow, we claim to have identified the worst-case canonical reactions---i.e., having the least $e(\alpha)/l(\alpha)$ or $k(\alpha)/l(\alpha)$ ratio---without proof.

We first consider the TBN AND gate introduced in prior work~\cite{DBLP:conf/dna/DotyRSTW17,DBLP:journals/tcs/BreikTHS19} and recently experimentally realized~\cite{Wang2024.09.13.612990}.
\Cref{fig:and_gate} (left) shows the desired functionality of the AND gate in which inputs $A$ and $B$ cooperate to produce $C$.
We are interested in bounding the concentration of $C$ (leak) that can be produced in the absence of inputs $A$ and $B$.
Phrased in our terminology, 
combinatorial TBN arguments have shown that $\sett{S} = \{ X, Y, Z\}$  is TBN-stability closed, as well as $\{ X, Y, Z, A\}$ and $\{ X, Y, Z, B\}$ where one of the two inputs is present.
This implies that any canonical reaction $\alpha$ producing $C$ has entropy loss $e(\alpha) \geq 1$.
However, such arguments do not directly connect entropy loss to leak concentrations,
justifying the need for a tool like our \Cref{cor:user-friendly}.

\begin{figure}[t]
    \centering
    \includegraphics[width=0.85\textwidth]{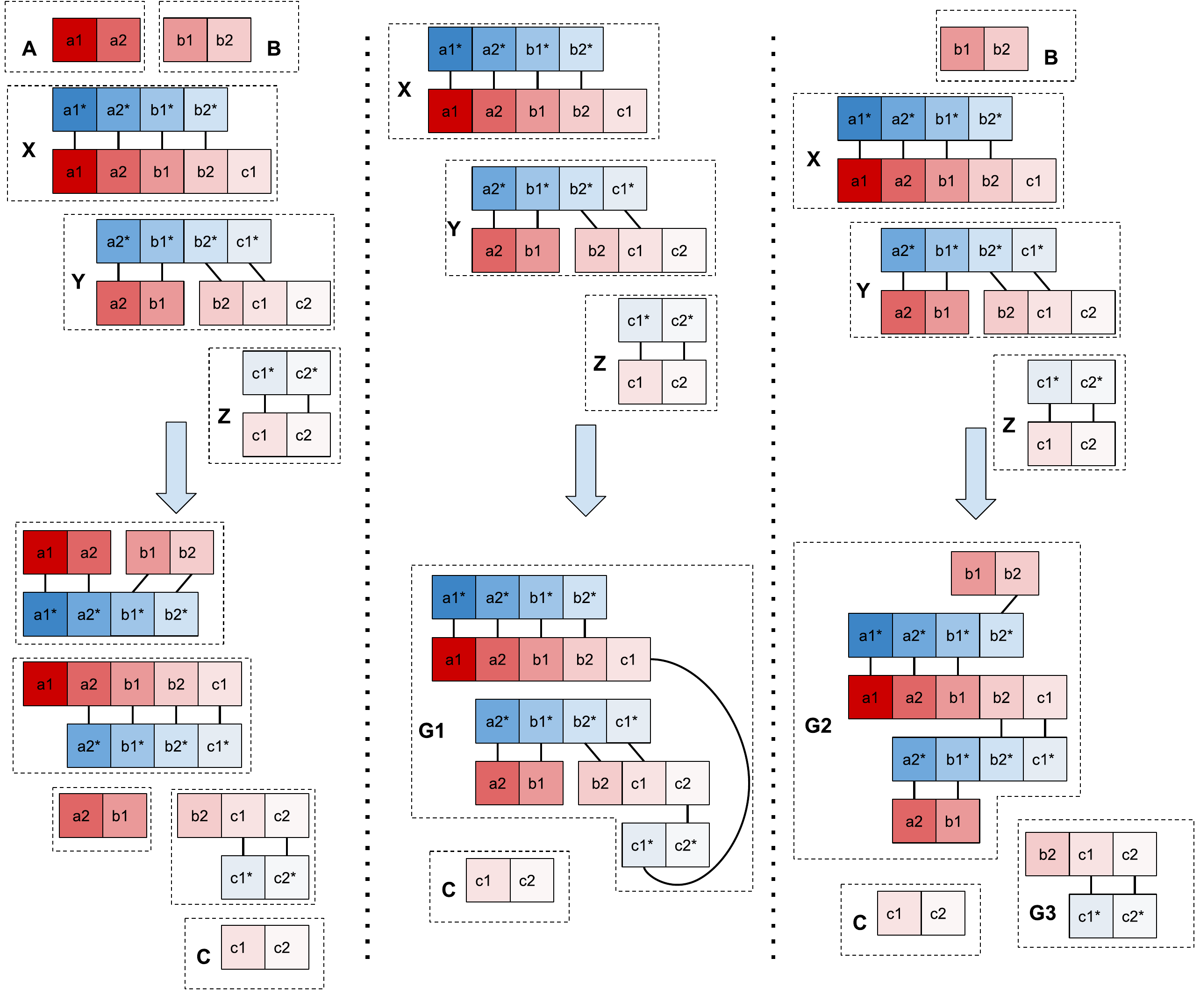}
    \caption{A TBN module implementing an AND gate. The presence of polymers $A$ and $B$ represents input 1, while their absence represents input 0. 
    (Left) The intended reaction pathway generates the output polymer $C$ when both $A$ and $B$ are present.
    (Middle) When both $A$ and $B$ are absent, the worst-case canonical reaction is $\alpha:X + Y + Z \to G_1 + C$ producing an erroneous output $C$ (leak). 
    The entropy loss and novelty of the reaction are $e(\alpha) = 1$ and $l(\alpha) = 2$, resulting in the concentration upper bound of $C$ of $c^{1.5}$ via \Cref{cor:user-friendly}.
    (Right) If only one input is present ($B$), then the worst-case canonical reaction is $\beta:B+X+Y+Z \to G_2 + G_3 + C$ with $e(\beta) = 1$ and $l(\beta) = 3$, yielding the upper bound concentration of $C$ of $c^{1.33}$.
    }
    \label{fig:and_gate}
\end{figure}

To apply our framework to the TBN AND gate, we use \Cref{cor:user-friendly}.
In the absence of both inputs, take $\sett{S} = \{ X, Y, Z\}$.
We claim that the worst-case canonical reaction $\alpha$ producing $C$ is 
$X + Y + Z \to G_1 + C$ shown in \Cref{fig:and_gate} (middle).
This reaction has entropy loss $e(\alpha) = 1$ and novelty $l(\alpha) = 2$ since $G_1$ and $C$ are outside of $\sett{S}$.
Thus for any base concentration $0 < c < 1$, there is an equilibrium with $[X] = [Y] = [Z] = c$ where the leak concentration is $[C] \leq c^{1.5}$.

Similarly, consider the case of having input $B$ but no input $A$. 
We take $\sett{S} = \{ X, Y, Z, B\}$ and claim\footnote{While we do not provide a formal proof that this reaction is worst-case, our confidence is based on the observation that we cannot increase the novelty further without increasing the entropy loss in proportion.
For example, we can combine two copies of reaction $\alpha$ to yield reaction $\alpha'$ but then $e(\alpha') = 2 e(\alpha)$ and $l(\alpha') = 2 l(\alpha)$, maintaining the same ratio.
} 
that the worst-case canonical reaction $\alpha$ producing $C$ is 
$X + Y + Z + B \to G_2 + G_3 + C$ as shown in \Cref{fig:and_gate} (right).
This reaction has entropy loss $e(\alpha) = 1$ and novelty $l(\alpha) = 3$ since $G_2$, $G_3$, and $C$ are outside of $\sett{S}$,
yielding the equilibrium leak concentration of $[C] \leq c^{1.33}$.
Our analysis thus provides concrete polynomial upper bounds on leak, 
consistent with the qualitative expectation of the TBN model that leak should become comparatively negligible in the limit of decreasing $c$.
The bound is smaller when both inputs are absent than when $B$ is present.

As the next example, we consider the ``leakless'' DNA strand displacement system previously theoretically and experimentally studied~\cite{DBLP:journals/pnas/WangB18}.
That work focuses on a family of ``translator'' modules that convert an input strand to an output strand of independent sequence. 
The family is parameterized by the redundancy parameter $N$ defined as the number of bound domains in each fuel polymer $F_i$, such that the number of domains in each signal $X_i$ is $N+1$.
Leak is expected to decrease with decreasing overall concentration (as for the AND gate), as well as with increasing redundancy $N$. 
The decrease of leak with increasing $N$ was confirmed by experiment, at least for small $N$.
The reactions of the translator with $N=3$ are shown in \cref{fig:tbn_translator}.

\begin{figure}[t]
    \centering
    \includegraphics[width=0.77\textwidth]{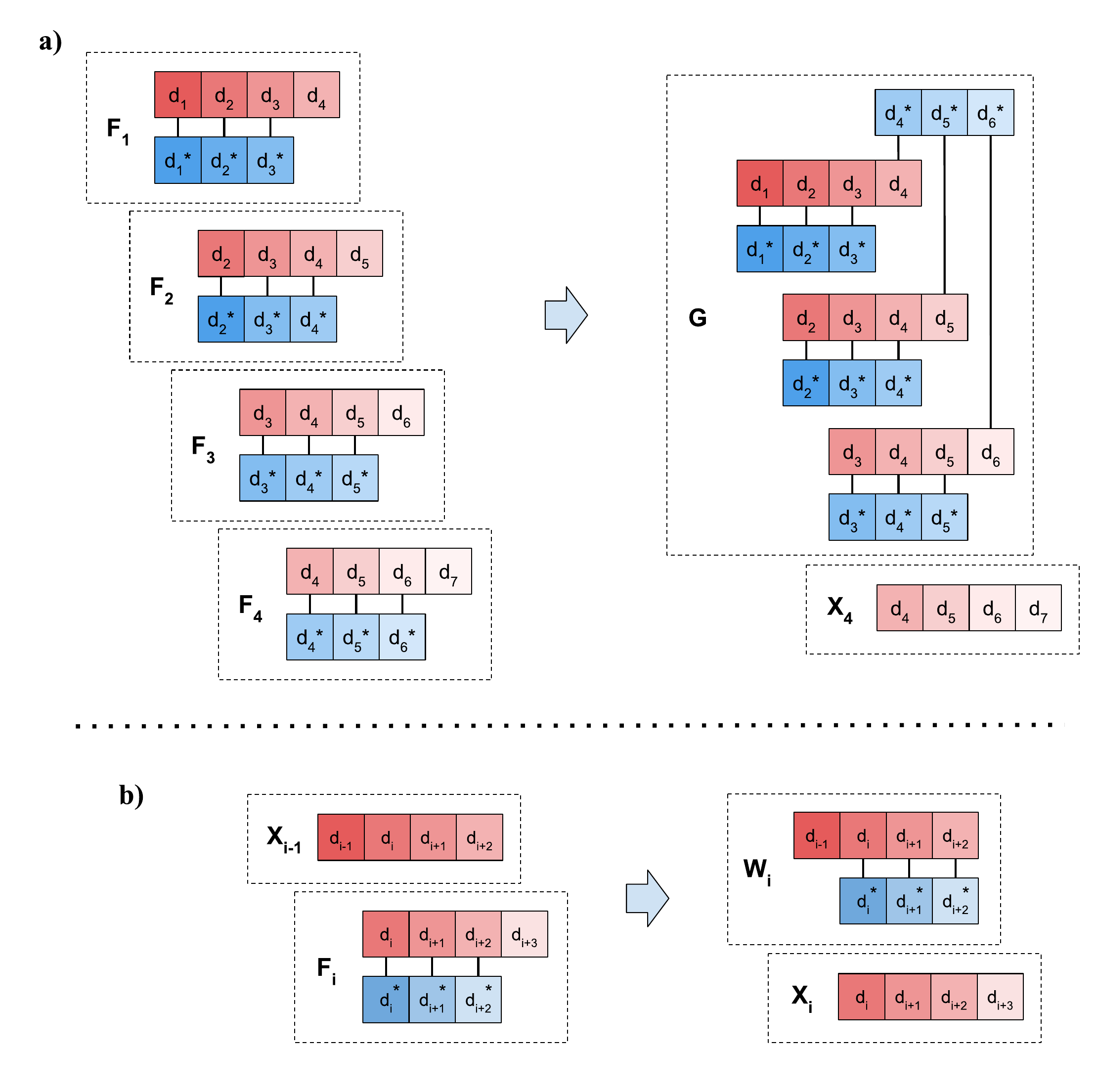}
    \caption{Translator cascade with redundancy parameter $N=3$. Polymer $X_0$ serves as input and polymer $X_n$ as output signal. 
    The leak pathway with $N + 1$ fuels coming together to generate a signal $X_{N+1}$ in the absence of input is described in (a). 
    Part (b) describes the intended reaction for each $i$; iterating them for $i=1,...,N+1$ produces $X_{N+1}$ given input $X_0$. 
    Note that if several translators are composed, then $X_{N+1}$ is the input to the downstream translator and once the leak signal is generated via the pathway shown in (a) it can propagate via intended reactions to the output $X_n$ of the last translator.
    }
    \label{fig:tbn_translator}
\end{figure}

To apply our framework, we consider the system without toeholds, driven solely by entropy;
with long domains alone, we are in the enthalpy neutral (athermic) regime.
Prior work focused on the case where the system was prepared in a state with only fuel polymers ($F_i$), all at equal concentration, and zero initial concentration of waste polymers ($W_i$). 
Let on-target set be $\sett S = \{F_i, W_i \mid i=1, \dots,n\}$ with uniform $\mu = 1$.
Rephrasing in our terminology, ref.~\cite{DBLP:conf/dna/ThachukC15} proved that any canonical reaction $\alpha: \multiset{M} \to \multiset{M'}$ where $X_i \in \multiset{M'}$ has entropy loss $e(\alpha) = N-1$,
and this fact was used to argue that by increasing $N$ we can make leak arbitrarily small.
We now show that considering novelty in addition to entropy loss makes this argument problematic, and suggest an alternative parameter setting to arbitrarily decrease leak.

Consider a cascade of two translators.
Importantly, the number of fuels for a single translator module is $N+1$; i.e., $X_i$ and $X_j$ overlap in sequence for $i < j < i+N+1$ but are completely sequence-independent for $j = i+N+1$.
We claim that the worst-case leak pathway $\alpha$ is where the first translator leaks resulting in the triggering of the second translator.
This pathway generates $l(\alpha) = N+3$ new off-target polymers: 
$1$ for the leaked upstream translator (all fuels coming together), 
$N+1$ for the triggered fuels of the second translator, and $1$ for the output.
Thus $e(\alpha)/l(\alpha) = (N-1)/(N+3)$.
By \cref{cor:user-friendly}, the concentration of the leak product in the uniform setting is bounded above by $c^{4/3}$ for $N=3$.
The bound tightens to $c^2$ with increasing $N$; however, it does not get arbitrarily small.
This suggests that we do not decrease equilibrium leak concentration arbitrarily by increasing ``redundancy'' $N$ because while the entropy loss increases, the novelty increases as well.

To arbitrarily decrease leak, we propose to use positive initial concentrations of the waste polymers $W_i$. 
The fuel polymers are denoted $F_1, ..., F_n$, and the waste polymers produced after translator triggering are $W_1, ..., W_n$. 
Let $X_0$ represent the input and $X_n$ represent the output. 
We have $n=2(N+1)$ for two translators composed together.
This cascade of length $n$ proceeds through the following reactions described in \cref{fig:tbn_translator}:
\[
X_{i-1} + F_i \to X_i + W_i
\]
for $i=1,...,n$.

We first consider the triggered system (with-input) showing at least a constant fraction of the signal is propagated to the end as $X_n$.
We then focus on the case without input and bound the leak.
Note that this analysis utilizes \Cref{theorem:level-i} directly rather than \Cref{cor:user-friendly} because we will have different concentration exponents $\mu$ for $F_i$ and $W_i$ in $\sett{S}$.

For the triggered (with-input) system we define our on-target set as $\sett S = \{F_i, W_i, X_0, X_i \mid i=1, \dots,n\}$.
We assign concentrations to the on-target polymers as follows: 
All fuel concentrations are equal to $2c$, and all waste concentrations are equal to $c$. These concentrations correspond to concentration exponents $\cex{F_i}=1 + \log_c 2$ and $\cex{W_i}=1$.
Let the concentration of the output in the final layer $X_n$ be $y$ and assign balancing concentrations of the other $X_i$. 
Since $[F_i]/[W_i] = 2$, we have $\sum_{i=0}^{n} [X_i] < 2y$, meaning that more than  half of the total signal ($X_i$) is at the output layer.

Now, we investigate the system without input. Let $\sett{S} = \{F_i, W_i \mid i = 1, \ldots, n\}$. 
For the situation to properly correspond to the with-input case, 
we need to ensure that all monomer concentrations are the same between the two cases, except removing the monomer corresponding to input $X_0$.
Rather than thinking about specific monomers, we start in the with-input case and conceptually run reactions $X_i + W_i \to X_{i-1} + F_i$ to completely push all $X_i$ to $X_0$, and then remove $X_0$ from the system.\footnote{More precisely, this ensures that the concentrations of monomers making up on-target polymers are the same between the with-input and without-input case (other than $X_0$).
}
Since the total amount of all $X_i$ is less than $2y$ in the with-input case, this results in: $[F_i] < 2c + 2y$ and $[W_i] > c - 2y$.
These correspond to 
$\cex{F_i} > \log_c(2c + 2y)$ and $\cex{W_i} < \log_c(c - 2y)$.

Recall that redundancy $N$ results in entropy penalty $N-1$.
We claim that the reaction with the smallest imbalance-novelty ratio (i.e., worst-case) is reaction $\beta$: 
\[
F_1 + \cdots + F_n \; \to \; G + W_{N+2} + \cdots + W_n + X_n,
\]
where $G$ is the ``large polymer'' formed after all $F_1, \ldots, F_N$ displace the top strand from $F_{N+1}$, and $X_n$ is the leak output.
The imbalance of this reaction is:
\begin{equation} \label{eqn:imbalance-bound}
k(\beta) = \sum_{i=1}^n \cex{F_i} - \sum_{i=N+2}^n \cex{W_i} > n \log_c (2c+2y) - \frac{n}{2} \log_c (c-2y).
\end{equation}

We can ensure that $k(\beta)$ is at least a constant fraction of $n$ for small enough $c$.
For example, if we let $y \leq c / 4$, then $k(\beta) \geq n/4$ for any $c \leq 0.0064$.
The novelty is independent of $n$: $l(\beta) = 2$ since $G$ and $X_n$ are not in $\sett{S}$.
Therefore, the imbalance-novelty ratio $k(\beta)/l(\beta)$ of the worst case reaction is at least $n/8$, which increases linearly with $N$ (recall $n=2(N+1)$ for two translators composed together).
Applying \Cref{theorem:level-i} leads to leak concentration of at most $c^{n/8} = c^{1/4+N/4}$. 
This upper bound\footnote{Note that this upper bound is loose because of  inequalities such as \Cref{eqn:imbalance-bound}.}
implies smaller-than $c$ concentration of leak for $N \geq 4$, with the leak exponentially decreasing for larger $N$.

To summarize, by increasing redundancy $N$ in the appropriate regime, we maintain the property that a constant fraction of the input is converted to output in the with-input case, while arbitrarily (exponentially in $N$) decreasing leak in the without-input case.

\section{Discussion}
\label{sec:discussion}

Our results suggest a few important directions for future work.
Given the central role of worst-case canonical reactions---i.e., canonical reactions with the lowest imbalance-novelty ratio (for \Cref{alg1} and \Cref{theorem:level-i}) or entropy loss-novelty ratio (\Cref{cor:user-friendly})---it is important to develop formal techniques to prove that a given canonical reaction is indeed  worst-case overall, at least or for a particular off-target polymer (for \Cref{sec:bounding-framework}). 
{Note that while combinatorial techniques in prior work in the TBN model have focused on proving entropy loss, more work is needed to study the ratio directly, making our framework more easily applicable.}
While we believe the canonical reactions highlighted in \Cref{sec:Examples} are indeed worst-case, the argument is informal.

Another promising avenue of research is to establish a more direct link between a polymer’s monomer composition and its equilibrium concentration. 
Our current framework is effectively reaction-centric, inferring concentrations based on how polymers transform into one another. 
An alternative approach could be to derive concentration bounds directly from the structural properties of the off-target polymers, such as their size (monomer count) or degree of overlap with one another (multiset difference).
Nonetheless, we hope that a variety of structure-based results could be proven based on a reduction to our canonical reaction framework.

Finally, this work has focused exclusively on the athermic case, where all molecular interactions are enthalpically neutral.  
While this is a reasonable and useful abstraction for systems with strong, saturated bonds,
many real-world molecular systems, including many popular in DNA molecular programming, involve a range of binding strengths and enthalpic effects (e.g., from toehold binding). 
Extending our algorithmic framework to incorporate user-specified $\Delta G$'s for each polymer could significantly broaden its applicability,
and although this would complicate our algorithm, we do not anticipate any insurmountable difficulties.

\bibliography{main}

\appendix

%
%

\section{Hilbert Basis Implementation}\label{subsec:Hilbert}
This section recalls some results on Hilbert bases and explains how to use the Hilbert basis in \cref{alg1}.
While this section makes the main theorem mathematically rigorous,
most readers can safely skip this section.

The main contents of this section are (1) the termination of \cref{alg1} and (2) the use of minimum (versus infimum) in \cref{def:ratio}. 
We address both concerns by showing that, although there are infinitely many canonical reactions, we can always restrict our attention to a finite subset of them in our analysis and algorithm.

For integral vectors $\vec v_1,...,\vec v_m \in \Z^n$, the set $C= \{\sum_{i=1}^m b_i \vec v_i: b_1,...,b_m \ge 0\}$ is called a (rational polyhedral) cone, which is also known to be described by a system of inequalities $C=\{\vec v: \vec B\cdot \vec v\le 0\}$ for some matrix $\vec B\in \Z^{l\times n}$. 
It is known that the set $C \cap \Z^n$ has a \emph{finite} subset $\mathcal H(C)=\{\vec h_1 ,..., \vec h_t\}$, called \emph{Hilbert basis} of $C$, that generates $C\cap \Z^n$ with non-negative integer coefficients, that is, for any $\vec v \in C \cap \Z^n$ there are $a_1,...,a_t\in \N$  such that 
$
\vec v = \sum_{i=1}^t a_i \vec h_i.
$

The set of canonical reactions $\Lambda$ can be precisely described in terms of a rational polyhedral cone. 
For a canonical reaction $\alpha:\multiset{M_1} \to \multiset{M_2}$, we define a vector 
\[\vec v_\alpha =(\multiset{M_1}[P]-\multiset{M_2}[P])_{P\in \polymers}\in \Z^{|\polymers|}\]
capturing the stoichiometric change of concentrations due to reaction $\alpha$.
Note that $\multiset{M_1}[P']-\multiset{M_2}[P'] = -\multiset{M_2}[P'] \le 0$ for $P'\notin \sett S$, thus
$\vec v_\alpha \cdot \vec e_{P'} \le 0 \text{ for }\vec e_{P'} = (\delta_{PP'})_{P \in \polymers}$
for the delta function $\delta_{ij}=1$ iff $i=j$.
\cref{def:reconfiguration} ensures that this vector must satisfy the condition $\vec A\cdot  \vec v_\alpha = 0$.
Combining these, the cone
\[
C^\sett S = \{\vec v_\alpha \in \R^{|\polymers|}:\vec A \cdot \vec v_\alpha \ge 0  \text{ and } \vec A \cdot \vec v_\alpha \le 0 \text{ and } \vec v_\alpha \cdot \vec e_{P} \le 0 \text{ for all }P\notin \sett S\}
\]
characterizes the canonical relations: $C^\sett S \cap \Z^{|\polymers|}$ is the set of vectors $\vec v_\alpha$.
Therefore, there exists a finite set of canonical reactions $H$ that corresponds to the Hilbert basis $\mathcal H(C^\sett S)$.

The following lemma implies that for the purposes of our algorithm and analysis, we can focus on the Hilbert basis $H$ of canonical reactions.
\begin{lemma}\label{lemma:canonical-merging}
    Let $\alpha:\multiset{M_1} \to \multiset{M_2}$ and $\beta:\multiset{N_1} \to \multiset{N_2}$ be canonical reactions with $l_i(\alpha),l_i(\beta)\neq 0$. Then, for any $a,b \in \N$, it holds that
    \begin{equation}\label{eqn:mediant}
    \frac{k_i(a\cdot \alpha + b\cdot \beta)}{l_i(a\cdot \alpha + b\cdot \beta)} \ge \min\left(
    \frac{k_i(\alpha)}{l_i(\alpha)},\frac{k_i(\beta)}{l_i(\beta)}
    \right).
    \end{equation}
    The equality holds only when $a=0$, $b=0$, or $k_i(\alpha)/l_i(\alpha) = k_i(\beta)/l_i(\beta)$.
\end{lemma}
\begin{proof}
It is not hard to see that $k_i$ and $l_i$ are linear, i.e., $k_i(a \cdot \alpha + b \cdot \beta) = a\cdot k_i(\alpha) + b\cdot k_i(\beta)$ and $l_i(a \cdot \alpha + b \cdot \beta) = a\cdot l_i(\alpha) + b\cdot l_i(\beta)$. Then, \cref{eqn:mediant} is identical to the mediant inequality, which states that for $p,q,r,s\ge0$ with $q,s\neq 0$ it holds that $\min(p/q,r/s) \le (p+r)/(q+s).$
\end{proof}
Consider a canonical reaction $\alpha=
a_1 \cdot \eta_1 + ... + a_t \cdot \eta_t$ for $a_1,...,a_t >0$ and $\eta_1,...,\eta_t \in H$. 
If $\alpha$ is $i$-levelizing, then the equality $k_i(\alpha)/l_i(\alpha)=\mu_i=\min_{j=1,..,t}(k_i(\eta_j)/l_i(\eta_j))$ must hold, and the equality condition above ensures that $\mu_i = k_i(\eta_j)/l_i(\eta_j)$ for all $j=1,...,t$. 
In other words, the set of $i$-levelizing reactions is
\[
\left\{\sum_{j=1}^t a_j \cdot \eta_j: a_1,...,a_t \in \N\right\} \text{ where }\{\eta_1,...,\eta_t\}\text{ is the set of $i$-levelizing reactions in }H.
\]
This allows us to inspect the minimum over the finite set $H$ instead of $\Lambda$ in \cref{def:ratio}.
Similarly, as the Hilbert basis can be computed in finite time and implemented in \cite{Normaliz},
we can run \cref{alg1} with guaranteed termination by computing the minimum over $H$.

\section{Detailed Balance and Equilibrium}\label{appndx:equilibrium}

Recall that $\vec{A} \in \N^{|\monomers| \times |\polymers|}$ is the matrix such that each entry $A_{ij}$ specifies the number of monomers of type $i$ in polymer $j$, and that $\vec A \cdot \vec x =\vec {x^0}$ (\Cref{eqn:constraints}) captures the mass-conservation constraint.

\begin{theorem}
    Let $\vec {x^0} \in (0,1)^{\monomers}$ be a fixed vector of monomer concentrations.
    If all reactions are balanced at the configuration $\vec x \in (0,1)^{\polymers}$ of polymer concentrations,
    then the cost function $g(\vec x)$ is minimum subject to $\vec A \cdot \vec x =\vec {x^0}$.
\end{theorem}
\begin{proof}
    (Sketch) The function $g$ is strictly convex since its Hessian $H$ is positive definite (specifically diagonal with $H_{jj} = 1/x_j > 0$).
    Strict convexity of $g$ implies that the local minimum of $g$ becomes the unique (global) minimum.

We associate a vector $\vec{v}_\alpha \in \Z^{|\polymers|}$ with every reaction $\alpha$, capturing the net stoichiometric effect of reaction $\alpha$.\footnote{This is the same vector $\vec{v}_\alpha$ as defined in \Cref{subsec:Hilbert}. 
For example, $(1,-1,0,...)$ corresponds to $X_1 \to X_2$ for $\Psi = \{X_1,X_2,...\}$.}
It is straightforward to show that the function $g$ along with the direction of $\vec{v}_\alpha$ has zero derivative at $\vec x$ if and only if the reaction $\alpha$ is balanced at $\vec x$. More explicitly, for $\alpha:\multiset{M_1}\to\multiset{M_2}$, the $P$-th entry of the vector $\vec v_\alpha$ is $(\multiset{M_1}[P]-\multiset{M_2}[P])$.
The directional derivative
$D_{\vec v_\alpha} g(\vec x) = \sum_P (\multiset{M_1}[P]-\multiset{M_2}[P]) \log x_P = 0$ implies that $\prod_{P\in \multiset{M_1}} {x_P}^{\multiset{M_1}[P]} = \prod_{P\in \multiset{M_2}} {x_P}^{\multiset{M_2}[P]}$ holds, which means the reaction $\alpha$ is balanced.

The set $\{\vec x: \vec A \cdot \vec x = 0\}$ is spanned by the vectors $\vec{v}_\alpha$ for all reactions by \Cref{def:reconfiguration}. 
Therefore, if all reactions are balanced at $\vec x$, then any directional derivative at $\vec x$ vanishes and $\vec x$ is a critical point, which is the unique minimum.
\end{proof}

\end{document}